\documentclass[english]{article}
\usepackage[T1]{fontenc}
\usepackage[latin9]{inputenc}
\usepackage{color}
\usepackage{refstyle}
\usepackage{float}
\usepackage{calc}
\usepackage{algorithm2e}
\usepackage{amsthm}
\usepackage{amsmath}
\usepackage{amssymb}
\usepackage{graphicx}
\usepackage{nameref}

\makeatletter

\AtBeginDocument{\providecommand\defref[1]{\ref{def:#1}}}
\AtBeginDocument{\providecommand\thmref[1]{\ref{thm:#1}}}
\AtBeginDocument{\providecommand\eqref[1]{\ref{eq:#1}}}
\RS@ifundefined{subref}
  {\def\RSsubtxt{section~}\newref{sub}{name = \RSsubtxt}}
  {}
\RS@ifundefined{thmref}
  {\def\RSthmtxt{theorem~}\newref{thm}{name = \RSthmtxt}}
  {}
\RS@ifundefined{lemref}
  {\def\RSlemtxt{lemma~}\newref{lem}{name = \RSlemtxt}}
  {}

\numberwithin{equation}{section}
\theoremstyle{plain}
\newtheorem{thm}{\protect\theoremname}
  \theoremstyle{definition}
  \newtheorem{defn}[thm]{\protect\definitionname}

\usepackage{calrsfs}
\usepackage{hyperref}
\hypersetup{hidelinks,colorlinks = false, allcolors=., citebordercolor = [rgb]{0,0,0}, linkbordercolor = [rgb]{0,0,0}, urlbordercolor = [rgb]{0,0,0}}

\makeatother

\usepackage{babel}
  \providecommand{\definitionname}{Definition}
\providecommand{\theoremname}{Theorem}

\begin{document}

\title{Truthful and Faithful Monetary Policy for a Stablecoin Conducted
by a Decentralised, Encrypted Artificial Intelligence}

\author{David Cerezo Sánchez\\
{\small{}david@calctopia.com}}
\maketitle
\begin{abstract}
The Holy Grail of a decentralised stablecoin is achieved on rigorous
mathematical frameworks, obtaining multiple advantageous proofs: stability,
convergence, truthfulness, faithfulness, and malicious-security. These
properties could only be attained by the novel and interdisciplinary
combination of previously unrelated fields: model predictive control,
deep learning, alternating direction method of multipliers (consensus-ADMM),
mechanism design, secure multi-party computation, and zero-knowledge
proofs. For the first time, this paper proves:

- the feasibility of decentralising the central bank while securely
preserving its independence in a decentralised computation setting

- the benefits for price stability of combining mechanism design,
provable security, and control theory, unlike the heuristics of previous
stablecoins

- the implementation of complex monetary policies on a stablecoin,
equivalent to ones used by central banks and beyond the current fixed
rules of cryptocurrencies that hinder their price stability

- methods to circumvent the impossibilities of Guaranteed Output Delivery
(G.O.D.) and fairness: standing on truthfulness and faithfulness,
we reach G.O.D. and fairness under the assumption of rational parties

As a corollary, a decentralised artificial intelligence is able to
conduct the monetary policy of a stablecoin, minimising human intervention.
\end{abstract}

\section{Introduction}

The Holy Grail of a stablecoin\cite{holyGrailCrypto}, an asset with
all the benefits of decentralisation but none of the volatility, remains
the most elusive single-horned creature of the cryptocurrency market.
In fact, price stability is the most wanted feature of a cryptocurrency:
in a recent survey\cite{privacyMoney}, hedging against depreciation
risk (i.e., price stability) was the most important attribute and
it has a much higher feature than anonymity (40\% vs. 1\%) or illiquidity
risk; however, subjects of the survey assigned to the anonymous medium-of-payment
a value on average only 1.44\% higher than to the non-anonymous medium-of-payment.

In monetary economics, monetary policy rules refer to a set of rule
of thumb that the central bank is committed to, so it can maintain
the price stability of a currency (Taylor rule, McCallum rule, inflation
targeting, fixed exchange rate targeting, nominal income targeting,
etc). However, the fixed rules for the emission of most cryptocurrencies\cite{cryptoMonetaryPolicies}
cannot maintain price stability: the inflexibility of their emission
rules and their inelasticity of supply provoke part of the high volatility
of the cryptocurrency market; their lack of good monetary rules preclude
their wide used as money\cite{bitcoinMonetaryRule} as they lack clear
a clear focus on monetary equilibrium; instead, they feature technical
rules for stabilising the difficulty of mining\cite{bitcoinDifficulty},
but not monetary rules. Stablecoins\cite{surveyStablecoin,ECBstablecoin,stabilizationCryptocurrencies}
were born to explicitly solve the volatility problem of cryptocurrencies:
however, their current formulation relies on heuristics\cite{dai,terraMoney,nubits,powStability,powPriceStabilization}
without a general mathematical framework within which advantageous
properties can be mathematically proven such as stability and convergence.
Stablecoins lacking stability regimes and/or convergence guarantees
suffer from the instabilities of unstable domains and deleveraging
spirals that cause illiquidity during crises\cite{instabilityBlockchain}:
these shortcomings cause price volatility, making cryptocurrencies
unusable as short-term stores of value and means of payment, increasing
barriers to adoption.

This paper introduces the novel combination of multiple mathematical
frameworks in order to design a decentralised stablecoin by inheriting
multiple useful properties of said frameworks: stability, convergence,
truthfulness, faithfulness, and malicious-security.

\paragraph{Contributions}

The main and novel contributions are:
\begin{itemize}
\item first formal treatment of decentralised stablecoin within which multiple
mathematical properties can be proven: stability, convergence, truthfulness,
faithfulness, and malicious-security.
\item dynamical models of economic systems: currency prediction with deep
learning, and stabilisation and emission of stablecoins.
\item decomposition of Model Predictive Controllers with consensus-ADMM
for their implementation in decentralised networks (i.e., blockchains).
\item protection against malicious adversaries in said decentralised networks.
\item from mechanism design, proofs to guarantee truthfulness for all the
parties involved and faithfulness of the execution for the decentralised
implementation.
\end{itemize}

\section{Related}

Previous cryptocurrencies with a controlled money suppy similar to
a central bank currency were centralised\cite{centrallyBankedCryptocurrencies,implementingRSCoin,UFCBCoin,cryptoeprint:2018:412}:
for first time, this paper solves the decentralisation of the monetary
policy, achieving a fully decentralised cryptocurrency when combined
with a public permissionless blockchain.

Most stablecoins are centralised: the few ones that are decentralised
(e.g., \cite{dai}), rely on heuristics without a general mathematical
framework within which advantageous properties can be mathematically
proven such as stability and convergence.

\section{Background}

This section provides a brief introduction to the main technologies
of the decentralised stablecoin: blockchains, model predictive control,
alternating direction method of multipliers (ADMM), mechanism design,
secure multi-party computation, and zero-knowledge proofs. A high-level
and conceptual rendering of the interrelationship between these techniques
can be found in Figure \ref{fig:Combination-techniques}.

\paragraph{Blockchains}

A blockchain is a distributed ledger that stores a growing list of
unmodifiable records called blocks that are linked to previous blocks.
Blockchains can be used to make online secure transactions, authenticated
by the collaboration of the P2P nodes allowing participants to verify
and audit transactions. Blockchains can be classified according to
their openness. Open, permissionless networks don't have access controls
and reach decentralised consensus through costly Proof-of-Work calculations
over the most recently appended data by miners. Permissioned blockchains
have identity systems to limit participation and do not rely on Proofs-of-Work.
Blockchain-based smart contracts are computer programs executed by
the nodes and implementing self-enforced contracts. They are usually
executed by all or many nodes (\textit{on-chain smart contracts}),
thus their code must be designed to minimise execution costs. Lately,
off-chain smart contracts frameworks are being developed that allow
the execution of more complex computational processes.

\paragraph{Model Predictive Control}

Advanced method of process control including constraint satisfaction:
a dynamical model of a system is used to predict the future evolution
of state trajectories while bounding the input to an admissible set
of values determined by a set of constraints, in order to optimise
the control signal and account for possible violation of the state
trajectories; at every time step, the optimal sequence over $N$ steps
in determined but only the first element is implemented. Model Predictive
Control is widely used in industrial settings, and its large literature
contains proofs of feasibility, stability, convergence, robustness
and many other useful properties that could be reused in many other
settings.

In this paper, multiple dynamic systems for Model Predictive Control
will be introduced: \nameref{sub:Model-AlgorithmicStablecoin} and
\nameref{sub:Economic-Model-CollaterisedStablecoin}; \nameref{sub:Economic-Model-CentralBankedCurrency};
\nameref{sec:Decentralised-Currency-Prediction}; \nameref{sec:Distributed-Implementation-Stablecoins};
and \nameref{sec:Auction-Mechanism-Issuing-Stablecoins}.

\paragraph{Alternating Direction Method of Multipliers (ADMM)}

Class of algorithms to solve distributed convex optimisation problems
by breaking them into smaller pieces, and distributing between multiple
parties\cite{ADMMdistrStats}. Itself a variant of the augmented Lagrangian
methods that use partial updates for the dual variable, it requires
exchanges of information between neighbors for every iteration until
converging to the result.

In this paper, multiple optimisation problems expressed in Model Predictive
Control will be decomposed with ADMM techniques in order to decentralise
their computation between multiple parties: \nameref{sec:Decentralised-Currency-Prediction};
\nameref{sec:Distributed-Implementation-Stablecoins}; and \nameref{sec:Distributed-Implementation-Auction-Mechanism}.

\paragraph{Mechanism Design}

Also called ``reverse game theory'', is a field of game theory and
economics in which a ``game designer'' chooses the game structure
where players act rationally and engineers incentives or economic
mechanisms, toward desired objectives pursuing a predetermined game's
outcome.

In this paper, parties truthfully report private information \ref{thm:strategyproof-auction}
(strategy-proofness) and faithfully execute a protocol (definition
\defref{(Faithful-Implementation)}, \thmref{faithful-distributed},
\thmref{faithful-stablecoins}).

\paragraph{Secure Multi-Party Computation}

Protocols for secure multi-party computation (MPC) enable multiple
parties to jointly compute a function over inputs without disclosing
said inputs (i.e., secure distributed computation). MPC protocols
usually aim to at least satisfy the conditions of inputs privacy (i.e.,
the only information that can be inferred about private inputs is
whatever can be inferred from the output of the function alone) and
correctness (adversarial parties should not be able to force honest
parties to output an incorrect result). Multiple security models are
available: semi-honest, where corrupted parties are passive adversaries
that do not deviate from the protocol; covert, where adversaries may
deviate arbitrarily from the protocol specification in an attempt
to cheat, but do not wish to be ``caught'' doing so ; and malicious
security, where corrupted parties may arbitrarily deviate from the
protocol. 

We utilise the framework SPDZ\cite{cryptoeprint:2011:535}, a multi-party
protocol with malicious security.

\paragraph{Zero-Knowledge Proofs}

Zero-knowledge proofs are proofs that prove that a certain statement
is true and nothing else, without revealing the prover's secret for
this statement. Additionally, zero-knowledge proofs of knowledge also
prove that the prover indeed knows the secret.

In this paper, zero-knowledge proofs are used to prove that a local
computation was executed correctly.

\noindent 
\begin{figure}[H]
\begin{centering}
\includegraphics[scale=0.4]{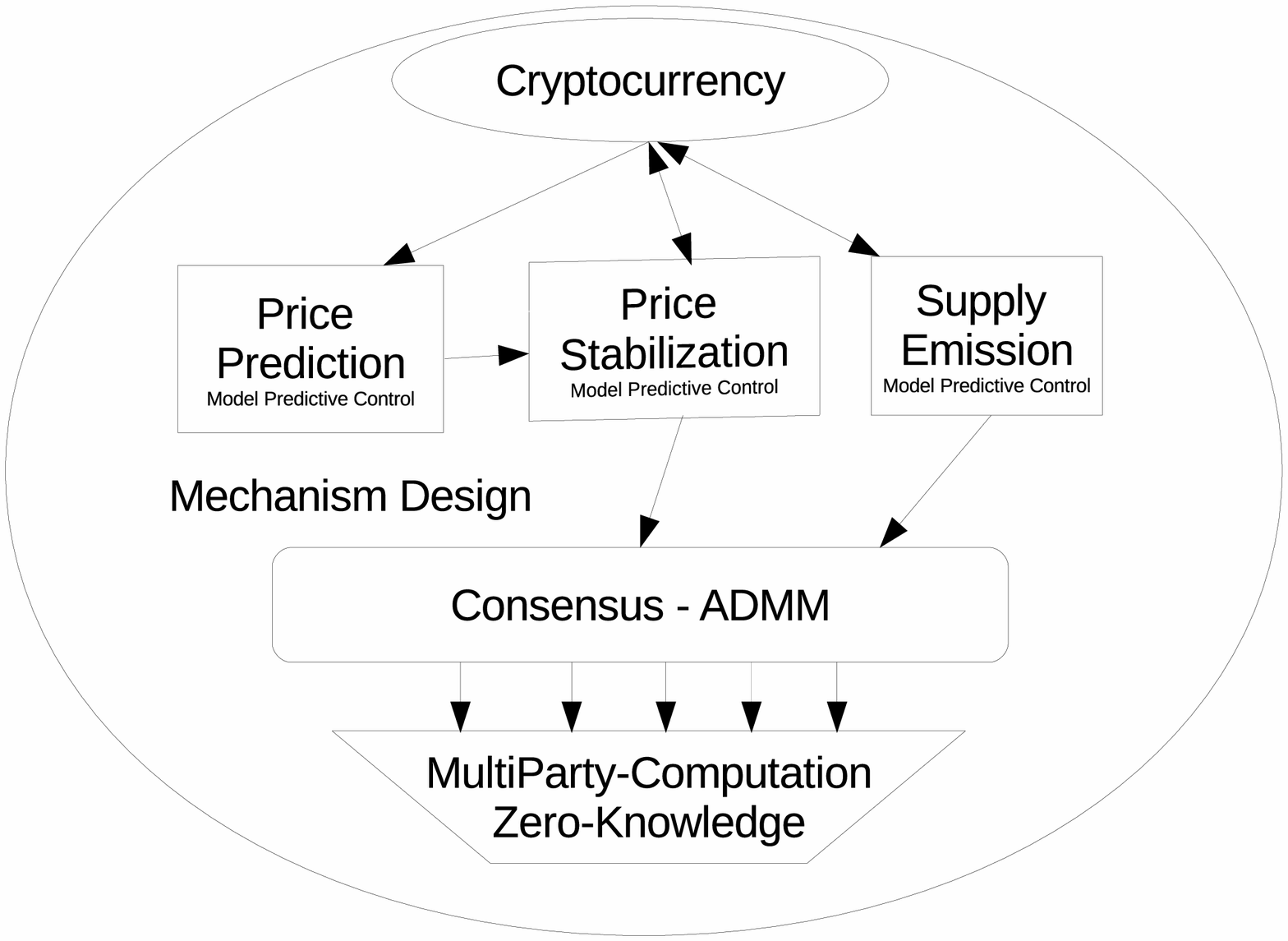}
\par\end{centering}

\protect\caption{\label{fig:Combination-techniques}High-level rendering of the combination
of techniques}
\end{figure}

\section{Economic Models}

We formalise\label{sec:basicModelCryptocurrency} a basic model of
a cryptocurrency%
\footnote{DISCLAIMER: the simplified models in the present paper are only for
illustrative purposes. Complex and parameterised models are needed
for real-world settings.%
} issuing variable block rewards and periodically auctioning a variable
amount of unissued coins from its uncapped and dynamically adjusted
supply: all these three variables are constantly adjusted by a dynamical
system using Stochastic Model Predictive Control in order to maintain
price stability (i.e., controlled variables).

Let $t\in T=\left\{ 1,\ldots,T\right\} $ denote the time slots used
by the blockchain. Let $S_{max}(t)$ denote the maximum supply of
a cryptocurrency, $S_{outstanding}(t)$ the supply that is visible
on-chain, $S_{initial}$ the initially issued supply by an initial
offering event (i.e., an initial auction) and $S_{unissued}(t)$ is
the amount of cryptocurrency yet to be issued. Then, we have:
\begin{eqnarray}
0 & \leq & S_{initial}\leq S_{outstanding}(t)\leq S_{max}(t),\forall t\in T,\label{eq:blockchainFirstConstraint}\\
S_{initial} & = & S_{outstanding}(1),\\
S_{max}(t) & = & S_{unissued}(t)+S_{outstanding}(t).
\end{eqnarray}
Periodically, miners are being rewarded for successfully processing
blocks with a variable amount of block rewards, $BR(t)$:
\begin{eqnarray}
S_{outstanding}(t+1) & = & S_{outstanding}(t)+BR(t),\forall t\in T,\\
S_{unissued}(t+1) & = & S_{unissued}(t)-BR(t),\\
0\leq & BR(t)\leq & BR_{max}.
\end{eqnarray}
Auctions are carried out to issue coins from the pool of $S_{unissued}(t)$
, each auction releasing a variable amount of auctioned coins, $AUC_{coins}(t)$,
with $x_{i}(t)$ denoting the amount of coins demanded by participant
$i$, $\forall t\in T$:
\begin{eqnarray}
S_{outstanding}(t+1) & = & S_{outstanding}(t)+AUC_{coins}(t),\\
S_{unissued}(t+1) & = & S_{unissued}(t)-AUC_{coins}(t),\\
0\leq & AUC_{coins}(t)\leq & AUC_{max},\label{eq:blockchainFirstHalfConstraint}\\
x_{i}^{min}(t)\leq & x_{i}\left(t\right) & \leq x_{i}^{max}\left(t\right),\label{eq:userConstraint}\\
\sum_{i}x_{i}(t) & -AUC_{coins}(t) & =0.\label{eq:auctionConstraint}
\end{eqnarray}
Let $P(t)$ denote the market price of a coin at time $t$ in a currency
(i.e., the number of cryptocurrency coins that one unit of currency
-EUR, JPY, USD- will buy at time $t$) and we adopt a geometric Brownian
motion model:
\begin{eqnarray}
\Delta P(t) & = & P(t+1)-P(t),\label{eq:blockchainSecondHalfConstraint}\\
dP(t) & = & \mu P(t)dt+\sigma P(t)dW_{t},
\end{eqnarray}
where $W_{t}$ is a Weiner process. Let $\left\{ S_{max}(t),BR(t),AUC_{coins}(t)\right\} $
be the controlled variables. Thus, in order to maintain price stability,
these controlled variables will expand when the price is increasing
and contract when the price is lowering:
\begin{eqnarray}
S_{max}(t) & \sim & S_{max}(t-1)\cdot\frac{P(t)}{P(t-1)},\\
BR(t) & \sim & BR(t-1)\cdot\frac{P(t)}{P(t-1)},\\
AUC_{coins}(t) & \sim & AUC_{coins}(t-1)\cdot\frac{P(t)}{P(t-1)}.\label{eq:blockchainLastConstraint}
\end{eqnarray}

\subsection{Economic Model for an Algorithmic Stablecoin\label{sub:Model-AlgorithmicStablecoin}}

Consider the stochastic linear state space system in the form
\begin{eqnarray}
x_{k+1} & = & Ax_{k}+Bu_{k}+w_{k}\label{eq:stochasticLinearSystem}\\
y_{k} & = & C_{y}x_{k}+v_{k}\\
z_{k} & = & C_{z}x_{k}
\end{eqnarray}
where $A,B,C_{y},C_{z}$ are state space matrices, $x_{k}\in\mathbb{R}^{n_{x}}$
is the state vector, $u_{k}\in\mathbb{R}^{n_{u}}$ is the input vector,
$y_{k}\in\mathbb{R}^{n_{y}}$ is the output vector, $z_{k}\in\mathbb{R}^{n_{z}}$
is the vector of controlled variables, $w_{k}\in\mathbb{R}^{n_{x}}$
is the noise vector of the process, and $v_{k}\in\mathbb{R}^{n_{y}}$
is the vector of measurement noise. Let $N$ be the length of the
prediction and receding horizon control and define the vectors
\[
N_{i}=\left\{ 0+i,1+i,\ldots,N-1+i\right\} 
\]
\begin{eqnarray*}
 & u=\left[\begin{array}{cccc}
u_{0}^{T} & u_{1}^{T} & \ldots & u_{N-1}^{T}\end{array}\right]^{T}, & x=\left[\begin{array}{cccc}
x_{1}^{T} & x_{2}^{T} & \ldots & x_{N}^{T}\end{array}\right]^{T},\\
 & z=\left[\begin{array}{cccc}
z_{1}^{T} & z_{2}^{T} & \ldots & z_{N}^{T}\end{array}\right]^{T}, & w=\left[\begin{array}{cccc}
w_{1}^{T} & w_{2}^{T} & \ldots & w_{N}^{T}\end{array}\right]^{T}
\end{eqnarray*}
Define the following exchange rate function measuring the cumulative
exchange rate between the price of a currency (e.g., EUR, JPY, USD)
and a stablecoin in the stochastic state space system \ref{eq:stochasticLinearSystem}
in the following $N$ time steps,
\begin{equation}
\begin{array}{c}
\psi_{xch}\left(u;\bar{x}_{0},w\right)=\left\{ \phi\left(u,x,z\right)\left|x_{0}=\bar{x}_{0},\right.\right.\\
\left.x_{k+1}=Ax_{k}+Bu_{k}+w_{k},z_{k+1}=Cx_{k+1},k\in N_{0}\right\} ,
\end{array}
\end{equation}
Let $P\left(t\right)$ be the spot price of the stablecoin cryptocurrency
denominated in a currency (e.g., EUR, JPY, USD). Then, the cumulative
exchange rate at time $t$ is
\begin{equation}
\phi\left(u,x,z\right)=\sum_{t=1}^{N}\left(P\left(t+1\right)\right)\label{eq:exchangeRateFunction}
\end{equation}
Following a criterion of social welfare maximisation, users and holders
of the stablecoin prefer to minimise the volatility of the exchange
rate, with the following equation describing the minimisation problem
\begin{equation}
\underset{\forall t}{\mbox{minimise }}\lambda E\left[\psi_{xch}\right]+\left(1-\lambda\right)\mbox{Var}\left[\psi_{xch}\right]\label{eq:minimizeAlgorithmicStablecoin}
\end{equation}
with $\lambda\in\left[0,1\right]$ determines the trade-off between
the expected exchange rate and the exchange rate variance.

\subsection{Economic Model for a Collaterised Stablecoin\label{sub:Economic-Model-CollaterisedStablecoin}}

We extend the basic model of a cryptocurrency (\ref{sec:basicModelCryptocurrency}),
with a reserve $R\left(t\right)$ backing every issued coin with $\lambda$
units of the reserve asset: for example, $\lambda=1$ for a 1 to 1
peg against a currency (e.g., EUR, JPY, USD), and $\lambda>1$ for
an overcollaterised stablecoin backed with other cryptocurrencies.
Then, we have:
\begin{eqnarray}
R\left(t\right) & = & \lambda\cdot S_{outstanding}\left(t\right),\\
0\leq & R\left(t\right) & \leq\lambda\cdot S_{max}\left(t\right),
\end{eqnarray}
In order to maintain price stability, $\lambda\left(t\right)$ could
also be a controlled variable that will increase when the price is
lowering and contract when the price is increasing:
\begin{eqnarray}
\lambda\left(t\right) & \sim & \lambda(t-1)\cdot\frac{P(t-1)}{P(t)},\\
1\leq & \lambda\left(t\right) & \leq\lambda_{max}.
\end{eqnarray}

\subsection{Economic Model for a Central-Banked Currency\label{sub:Economic-Model-CentralBankedCurrency}}

The framework and results of this paper could also be applied to the
monetary policy conducted by central banks, just by representing their
models in the framework of Model Predictive Control in a way similar
to the previous \nameref{sub:Model-AlgorithmicStablecoin}. 

The Taylor rule\cite{taylorRule} is an approximation of the responsiveness
of the nominal short-term interest rate $i_{t}$ as applied by the
central bank to changes in inflation $\pi$ and output $y$, according
to the following formula
\begin{equation}
i_{t}=\varphi_{y}\left(y_{t}-y^{*}\right)+\varphi_{\pi}\left(\pi_{t}-\pi^{*}\right)+\pi^{*}+r^{*}\label{eq:taylorRule}
\end{equation}
where a standard model describes the evolution of the economy
\begin{eqnarray}
\pi_{t+1} & = & \pi_{t}+\alpha y_{t}+e_{t+1}^{\pi},\label{eq:taylorModelInflation}\\
y_{t+1} & = & \rho y_{t}-\zeta\left(i_{t}-\pi_{t}\right)+e_{t+1}^{y},\label{eq:taylorModelOutput}
\end{eqnarray}
describing the dynamic relationship between the manipulated input
$i_{t}$ and the two controlled outputs $y_{t}$ and $\pi_{t}$. At
equilibrium, we obtain $i_{t}=i^{*}$, $\pi_{t}=\pi^{*}$, $y_{t}=0$
and $r^{*}=i^{*}-\pi^{*}$. Equations \ref{eq:taylorModelInflation}
and \ref{eq:taylorModelOutput} can be rewritten in the terms of deviation
variables from the equilibrium point, as
\begin{equation}
x_{t+1}=Ax_{t}+Bu_{t}+\epsilon_{t+1},\label{eq:stateSpaceTaylorModel}
\end{equation}
where
\begin{eqnarray}
x & = & \left[\begin{array}{cc}
y & -y^{*}\\
\pi & -\pi^{*}
\end{array}\right],u=\Delta i=i-i^{*},\epsilon=\left[\begin{array}{c}
e^{y}\\
e^{\pi}
\end{array}\right],\\
A & = & \left[\begin{array}{cc}
\rho & \zeta\\
\alpha & 1
\end{array}\right],\\
B & = & \left[\begin{array}{c}
-\zeta\\
0
\end{array}\right]
\end{eqnarray}
The cost function of the central bank is of the standard optimal control
form
\begin{equation}
\sum_{k=0}^{\alpha}\beta^{k}L\left(\hat{x}_{t+k\left|t\right.},u_{t+k\left|t\right.}\right)\label{eq:costFunctionCentralBank}
\end{equation}
where $\beta\in\left(0,1\right)$ is the discount factor, $\hat{x}_{t+k\left|t\right.}$
is the expected value of $x$ at time $t+x$ using all information
available at time $t$ and model \ref{eq:stateSpaceTaylorModel};
$u_{t+k\left|t\right.}$ is the input value at time $t+k$ decided
on at time $t$; and the $M$ function is usually defined as
\begin{equation}
M\left(\hat{x}_{t+k\left|t\right.},u_{t+k\left|t\right.}\right)=\hat{x}_{t+k\left|t\right.}^{T}Q\hat{x}_{t+k\left|t\right.}+R^{2}u_{t+k\left|t\right.}^{2}\label{eq:MfunctionCentralBank}
\end{equation}
with $R^{2}\geq0$ and $Q\succeq0$. The previous equations \ref{eq:costFunctionCentralBank}
and \ref{eq:MfunctionCentralBank} can be reformulated as an objective
for Model Predictive Control as
\begin{equation}
\underset{u}{\mbox{min}}\begin{array}{c}
\left\{ \sum_{k=0}^{N-1}\beta^{k}\left(\hat{x}_{t+k\left|t\right.}^{T}Q\hat{x}_{t+k\left|t\right.}+R^{2}u_{t+k\left|t\right.}^{2}+S^{2}\delta u_{t+k\left|t\right.}^{2}\right)\right.\\
\left.+\hat{x}_{t+N\left|t\right.}^{T}\beta^{N}\bar{Q}\hat{x}_{t+N\left|t\right.}+\beta^{N}S^{2}\delta u_{t+N\left|t\right.}^{2}\right\} 
\end{array}\label{eq:MPCTaylorRule}
\end{equation}
where
\begin{eqnarray}
Q & = & \left[\begin{array}{cc}
1-\lambda & 0\\
0 & \lambda
\end{array}\right]\succ0,0<\lambda<1,\\
u & = & \left[\Delta i_{t\left|t\right.}\ldots\Delta i_{t+N-1\left|t\right.}\right]^{T},\\
\delta u_{t+k\left|t\right.} & = & u_{t+k\left|t\right.}-u_{t+k-1\left|t\right.},k=0,\ldots,N\\
u_{t+k\left|t\right.} & \geq & -i^{*},k=0,\ldots,N-1,\label{eq:zero-lower-bound}\\
u_{t+k\left|t\right.} & = & u_{t+m-1\left|t\right.},k=m,\ldots,N-1,\\
\hat{x}_{t+k\left|t\right.}^{T} & = & \sum_{l=0}^{k-1}A^{l}Bu_{t+k-l-1\left|t\right.}+A^{k}x_{t},k=1,\ldots,N\label{eq:MPCTaylorRuleEnd}
\end{eqnarray}
with $\hat{x}_{t\left|t\right.}=x_{t}$ and the values of $1-\lambda$
and $\lambda$ determine the trade-off between the output gap and
inflation.

The decentralised implementation of the previous Model Predictive
Control \ref{eq:MPCTaylorRule} using the ADMM decomposition technique
is left as an exercise to the central banker.

\subsubsection{Closed-Loop Stability\label{sub:Closed-Loop-Stability}}

The following closed-loop structure is obtained from \ref{eq:taylorRule}
and \ref{eq:stateSpaceTaylorModel}:
\begin{equation}
x_{t+1}=A^{'}x_{t}+\epsilon_{t+1},
\end{equation}
where
\begin{equation}
A^{'}=A+Bc^{T}=\left[\begin{array}{cc}
\rho-\zeta\varphi_{y} & \zeta-\zeta\varphi_{\pi}\\
\alpha & 1
\end{array}\right]
\end{equation}

\begin{thm}
The Model Predictive Controller for the Taylor rule \ref{eq:MPCTaylorRule}-\ref{eq:MPCTaylorRuleEnd}
has closed-loop stability, if and only if,
\begin{eqnarray}
0.1\varphi_{\pi}-2.1 & <\varphi_{y}\\
 & \varphi_{y}< & 0.06\varphi_{\pi}+8.5,\\
\varphi_{\pi} & > & 1.
\end{eqnarray}
\end{thm}
\begin{proof}
The characteristic equation for the matrix $A^{'}$ is:
\begin{equation}
f\left(\mu\right)=\mu^{2}-\mu\left(\alpha\zeta-\zeta\varphi_{y}-\alpha\zeta\varphi_{\pi}+1+\rho\right)+\left(\rho-\zeta\varphi_{y}\right)
\end{equation}
where $\mu$ is an eigenvalue of matrix $A^{'}$. The closed-loop
system is stable when both eigenvalues of $A^{'}$ are inside the
unit disk (Jury\cite{JuryCriterion} and Routh\cite{RouthCriterion}-Hurtwiz\cite{HurwitzCriterion}
stability criteria), if and only if, 
\begin{eqnarray}
2+2\rho-2\zeta\varphi_{y}+\alpha\zeta\left(\varphi_{\pi}-1\right) & > & 0,\\
1-\rho+\zeta\varphi_{y}-\alpha\zeta\left(\varphi_{\pi}-1\right) & > & 0,\\
\alpha\zeta\left(\varphi_{\pi}-1\right) & > & 0.
\end{eqnarray}

\end{proof}
Similar stability results can be derived for the Model Predictive
Controllers of the \nameref{sub:Model-AlgorithmicStablecoin} and
the \nameref{sub:Economic-Model-CollaterisedStablecoin}.

\subsubsection{On Negative Interests}

The Model Predictive Controller for the Taylor rule (\ref{eq:MPCTaylorRule})-(\ref{eq:MPCTaylorRuleEnd})
includes a constraint for the zero lower bound on the interest rate,
\eqref{zero-lower-bound}:
\[
u_{t+k\left|t\right.}\geq-i^{*},k=0,\ldots,N-1.
\]
In case the central bank wants to implement negative interest rates,
said \eqref{zero-lower-bound} must be removed. A possible implementation
of negative interests for a cryptocurrency starts by considering coinage
epochs and then defining a depreciation rate for every coinage epoch
as time elapses. In the basic model of a cryptocurrency (\ref{sec:basicModelCryptocurrency}),
we could add the following equation:
\begin{eqnarray}
S_{outstanding}\left(t\right) & = & \sum_{t=0}^{T}\left(S_{minted}\left(t\right)-D_{T}\left(t\right)\right),\\
S_{initial} & = & S_{minted}(1)=S_{outstanding}(1),\\
0 & \leq & S_{initial}\leq S_{outstanding}(t)\leq S_{minted}\left(t\right)\leq S_{max}(t),\,\,\,\,\,\,\,
\end{eqnarray}
where $S_{minted}\left(t\right)$ is the amount of minted coins at
time $t$ and $D_{T}\left(t\right)$ is the depreciation of coins
minted at time $t$ evaluated at time $T$, for example,
\begin{eqnarray}
D_{T}\left(t\right) & = & \mbox{min}\left(\left(T-t\right)\cdot D_{rate}\cdot S_{minted}\left(t\right),S_{minted}\left(t\right)\right),\\
D_{rate} & = & 0.01,
\end{eqnarray}
for a 1\% depreciation rate for every coinage epoch since the first
epoch.

\section{Decentralised Prediction of Currency Prices through Deep Learning\label{sec:Decentralised-Currency-Prediction}}

As noted in previous publications about predicting markets using Stochastic
Model Predictive Control techniques\cite{mpcStockTrading}, this approach
is only justifiable only for consistent prediction of the direction
of price changes (i.e., sign changes): thus, it's a requisite to use
artificial intelligence techniques to predict price movements in order
to maintain price stability. Of course, price data can be shifted
by one sampling interval to the past, thereby making the economic
models independent of any predictive power: however, the correct formulation
is to use any potential good estimate of step-ahead prices as this
is the core of Stochastic Model Predictive Control. Therefore, the
exchange rate function \ref{eq:exchangeRateFunction} of the models
is formulated with one step-ahead prices ($P\left(t+1\right)$).

A neural network has $L$ layers, each defined by a linear operator
$W_{l}$ and a neural non-linear activation function $h_{l}$. A layer
computes and outputs the non-linear function:
\begin{equation}
a_{l}=h_{l}\left(W_{l}a_{l-1}\right)
\end{equation}
on input activations $a_{l-1}$. By nesting the layers, composite
functions are obtained, for example,
\begin{equation}
f\left(a_{0},W\right)=W_{4}\left(h_{3}\left(W_{3}\left(h_{2}\left(W_{2}h_{1}\left(W_{1}a_{0}\right)\right)\right)\right)\right)
\end{equation}
where the collection of weight matrices is $W=\left\{ W_{l}\right\} $.
Training a neural network for deep learning is the task of finding
the $W$ that matches the output activations $a_{L}$ to targets $y$,
given inputs $a_{0}$: it's equivalent to the following minimisation
problem, given loss function $l$,
\begin{equation}
\underset{W}{\mbox{minimise}}l\left(f\left(a_{0};W\right),y\right)\label{eq:minimizationNeural1}
\end{equation}
And this is equivalent to solving the following problem:
\begin{eqnarray}
\underset{\left\{ W_{l}\right\} ,\left\{ a_{l}\right\} ,\left\{ z_{l}\right\} }{\mbox{minimise}} & l\left(z_{L},y\right)\\
\mbox{subject to} & z_{l}= & W_{l}a_{l-1},\mbox{ for }l=1,2,\ldots,L,\\
 & a_{l}= & h_{l}\left(z_{l}\right),\mbox{ for }l=1,2,\ldots,L-1,
\end{eqnarray}
where a new variable stores the output of layer $l$, $z_{l}=W_{l}a_{l-1}$,
and the output of the link function is represented as a vector of
activations $a_{l}=h_{l}\left(z_{l}\right)$. By following the penalty
method, a ridge penalty function is added to obtain the following
unconstrained problem
\begin{equation}
\begin{array}{c}
\underset{\left\{ W_{l}\right\} ,\left\{ a_{l}\right\} ,\left\{ z_{l}\right\} }{\mbox{minimise}}\left\langle z_{L},\lambda\right\rangle +l\left(z_{L},y\right)+\beta_{L}\left\Vert z_{L}-W_{L}a_{L-1}\right\Vert ^{2}\\
+\sum_{l=1}^{L-1}\left[\beta_{l}\left\Vert z_{l}-W_{l}a_{l-1}\right\Vert ^{2}+\gamma_{l}\left\Vert a_{l}-h_{l}\left(z_{l}\right)\right\Vert ^{2}\right]
\end{array}\label{eq:minimizationNeuralExtended}
\end{equation}
where $\left\{ \gamma_{l}\right\} $ and $\left\{ \beta_{l}\right\} $
are constants controlling the weight of each constraint, and $\left\langle z_{L},\lambda\right\rangle $
is a Lagrange multiplier term. The advantage of the previous formulation
resides in that each sub-step has a simple closed-form solution with
only one variable, thus these sub-problems can be solved globally.

The update steps of each variable in the minimisation problem \ref{eq:minimizationNeuralExtended}
are considered as follows:
\begin{itemize}
\item To obtain $W_{l}$, each layer minimises $\left\Vert z_{l}-W_{l}a_{l-1}\right\Vert ^{2}$:
the solution of this least square problem is
\begin{equation}
W_{l}\leftarrow z_{l}a_{l+1}^{+}
\end{equation}
where $a_{l+1}^{+}$ is the pseudo-inverse of $a_{l+1}$.
\item To obtain $a_{l}$, another least-squares problem must be solved.
The solution is 
\begin{equation}
a_{l}\leftarrow\left(\beta_{l+1}W_{l+1}^{T}W_{l+1}+\gamma_{l}I\right)^{-1}\left(\beta_{l+1}W_{l+1}^{T}z_{l+1}+\gamma_{l}h_{l}\left(z_{l}\right)\right)
\end{equation}

\item The update for $z_{l}$ requires minimising
\[
\mbox{arg min}_{z}\gamma_{l}\left\Vert a_{l}-h_{l}\left(z\right)\right\Vert ^{2}+\beta_{l}\left\Vert z_{l}-W_{l}a_{l-1}\right\Vert ^{2}
\]

\item Finally, the update of the Lagrange multiplier is given by
\begin{equation}
\lambda\leftarrow\lambda+\beta_{L}\left(z_{L}-W_{L}a_{L-1}\right)
\end{equation}

\end{itemize}
All the previous steps are listed in the next Algorithm \ref{alg:ADMM-NeuralNets}:

\noindent \begin{flushleft}
\fbox{\begin{minipage}[t]{1\columnwidth}%
\begin{algorithm}[H]
\protect\caption{\label{alg:ADMM-NeuralNets}ADMM algorithm for Deep Learning}

\textbf{do}

~~\textbf{for }$l=1,2,\ldots,L-1$ \textbf{do}

~~~~$W_{l}\leftarrow z_{l}a_{l+1}^{+}$

~~~~$a_{l}\leftarrow\left(\beta_{l+1}W_{l+1}^{T}W_{l+1}+\gamma_{l}I\right)^{-1}\left(\beta_{l+1}W_{l+1}^{T}z_{l+1}+\gamma_{l}h_{l}\left(z_{l}\right)\right)$

~~~~$z_{l}\leftarrow\mbox{arg min}_{z}\left(\gamma_{l}\left\Vert a_{l}-h_{l}\left(z\right)\right\Vert ^{2}+\beta_{l}\left\Vert z_{l}-W_{l}a_{l-1}\right\Vert ^{2}\right)$

~~\textbf{end for}

~~$W_{L}\leftarrow z_{L}a_{L-1}^{+}$

~~$z_{l}\leftarrow\mbox{arg min}_{z}\left(l\left(z,y\right)+\left\langle z_{L},\lambda\right\rangle +\beta_{L}\left\Vert z-W_{L}a_{l-1}\right\Vert ^{2}\right)$

~~$\lambda\leftarrow\lambda+\beta_{L}\left(z_{L}-W_{L}a_{L-1}\right)$

\textbf{until} converged;
\end{algorithm}
\end{minipage}}
\par\end{flushleft}

\noindent \begin{flushleft}
Finally, note that more advanced methods for training neural networks
for deep learning have appeared in the literature\cite{diffLinerizedADMM,admmDeepLearning},
also considering their convergence.
\par\end{flushleft}

\section{Decentralised Stabilisation of Stablecoins\label{sec:Distributed-Implementation-Stablecoins}}

Following the \nameref{sub:Model-AlgorithmicStablecoin} and its minimisation
problem (\ref{eq:minimizeAlgorithmicStablecoin}), the expectation
of the exchange rate and the variance of the exchange rate are traded
off in a mean-variance Optimal Control Problem with the following
objective function
\begin{equation}
\psi=\lambda E_{w}\left[\psi_{xch}\right]+\left(1-\lambda\right)\mbox{Var}_{w}\left[\psi_{xch}\right]
\end{equation}
with $\lambda\in\left[0,1\right]$ determining the trade-off between
the expected exchange rate and the exchange rate variance. Estimates
of prices for the expected exchange rate, $E_{w}\left[\psi_{xch}\right]$,
and the variance, $\mbox{Var}_{w}\left[\psi_{xch}\right]$, are introduced
as follows
\begin{eqnarray}
E_{w}\left[\psi_{xch}\right] & \approx & \mu=\frac{1}{S}\sum_{i\in S}\psi_{xch}\left(u;\hat{x}_{0},w^{i}\right)\\
\mbox{Var}_{w}\left[\psi_{xch}\right] & \approx & s^{2}=\frac{1}{S-1}\sum_{i\in S}\left(\psi_{xch}\left(u;\hat{x}_{0},w^{i}\right)-\mu\right)^{2}
\end{eqnarray}
where $w^{i}$ is sampled from the distribution $w$ and $S$ is the
set of scenarios: when the number of scenarios is large, then
\begin{equation}
\psi\approx\tilde{\psi}=\lambda\mu+\left(1-\lambda\right)s^{2}\label{eq:largeScenario}
\end{equation}

The open-loop input trajectory is defined as the trajectory, $u^{*}\in U$,
that minimises (\ref{eq:largeScenario}), with $U$ being some input
constraint set. For the stochastic linear system (\ref{eq:stochasticLinearSystem}),
$u^{*}$ can be expressed as the solution to the following Optimal
Control Problem,
\begin{eqnarray}
 & \underset{\left\{ u^{j}\in U,x^{j},z^{j},\psi^{j}\right\} _{j=1}^{S},\mu}{\mbox{minimise}} & \lambda\mu+\tilde{\lambda}\sum_{j\in S}\left(\psi^{j}-\mu\right)^{2}\label{eq:optimalControlProblem}\\
\mbox{subject to} & \left(x^{i},u^{i},z^{i}\right)\in H\left(\hat{x_{0}},w^{i}\right), & i\in S,\\
 & \psi^{i}\geq\phi\left(u^{i},x^{i},z^{i}\right), & i\in S,\\
 & \mu=\frac{1}{S}\sum_{j\in S}\psi^{j},\\
 & u_{k}^{i}=u_{k}^{j}, & i,j\in S,j\in M\label{eq:optimalControlProblemEnd}
\end{eqnarray}
\[
\begin{aligned}\mbox{ where } & \tilde{\lambda}=\frac{1-\lambda}{S-1},\\
 & M=\left\{ 0,1,\ldots,M\right\} ,\\
 & M\leq N,\\
 & H\left(\hat{x}_{0},w\right)=\left\{ \left(x,z,u\right)\left|x_{0}=\hat{x}_{0},\right.\right.\\
 & x_{k+1}=Ax_{k}+Bu_{k}+w_{k},\\
 & \left.z_{k+1}=C_{z}x_{k+1},k\in N_{0}\right\} 
\end{aligned}
\]
The previous Optimal Control Problem (\ref{eq:optimalControlProblem})
is a convex optimisation problem when $U$ is a convex set and $\phi$
is a convex function: an ADMM-based decomposition algorithm for (\ref{eq:optimalControlProblem})
is presented below.

\subsection{ADMM Decomposition}

The Optimal Control Problem (\ref{eq:optimalControlProblem}) is re-written
as
\begin{eqnarray}
 & \underset{u\in\tilde{U},x,z,\psi,\mu}{\mbox{minimise}} & \lambda\mu+\tilde{\lambda}\psi^{T}\psi+S\tilde{\lambda}\mu^{2}-2\tilde{\lambda}\mu\mathbf{1}^{T}\psi,\label{eq:rewrittenOCP}\\
\mbox{subject to} & \tilde{A}x+\tilde{B}u+\tilde{w}=0,\\
 & z=\tilde{C}x,\\
 & \psi\geq\tilde{\phi}\left(u,x,z\right),\\
 & \mu=\mathbf{1}^{T}\psi/S,\\
 & \tilde{L}u=0,
\end{eqnarray}
where
\[
u=\left[\begin{array}{c}
u^{1}\\
u^{2}\\
\vdots\\
u^{S}
\end{array}\right],x=\left[\begin{array}{c}
x^{1}\\
x^{2}\\
\vdots\\
x^{S}
\end{array}\right],z=\left[\begin{array}{c}
z^{1}\\
z^{2}\\
\vdots\\
z^{S}
\end{array}\right],\psi=\left[\begin{array}{c}
\psi^{1}\\
\psi^{2}\\
\vdots\\
\psi^{S}
\end{array}\right],
\]
\[
\mathbf{1}=\left[\begin{array}{cccc}
1 & 1 & \ldots & 1\end{array}\right],
\]
\begin{eqnarray*}
\tilde{A}=\mbox{\textbf{blkdiag}}\left(\bar{A},\bar{A},\ldots,\bar{A}\right), & \tilde{B}=\mbox{\textbf{blkdiag}}\left(\bar{B},\bar{B},\ldots,\bar{B}\right), & \tilde{C}=\mbox{\textbf{blkdiag}}\left(\bar{C},\bar{C},\ldots,\bar{C}\right),\\
\bar{B}=\mbox{\textbf{blkdiag}}\left(B,B,\ldots,B\right), &  & \bar{C}=\mbox{\textbf{blkdiag}}\left(C,C,\ldots,C\right),
\end{eqnarray*}
\[
\bar{A}=\left[\begin{array}{cccc}
-I\\
A & -I\\
 & \ddots & \ddots\\
 &  & A & -I
\end{array}\right],\bar{w}^{i}=\left[\begin{array}{c}
w_{0}^{i}\\
w_{1}^{i}\\
\vdots\\
w_{N-1}^{i}
\end{array}\right]+\left[\begin{array}{c}
Ax_{0}\\
0\\
\vdots\\
0
\end{array}\right],i\in S,
\]
\begin{eqnarray*}
 & \tilde{w}=\left[\left(\bar{w}^{1}\right)^{T}\left(\bar{w}^{2}\right)^{T}\ldots\left(\bar{w}^{S}\right)^{T}\right]^{T},\\
 & \tilde{\phi}\left(u,x,z\right)=\left[\phi\left(u^{1},x^{1},z^{2}\right)\ldots\phi\left(u^{S},x^{S},z^{S}\right)\right]^{T},\\
 & \tilde{L}=\left[\begin{array}{ccccc}
L & -L\\
 & L & -L\\
 &  & \ddots & \ddots\\
 &  &  & L & -L
\end{array}\right],\\
 & L=\left[\begin{array}{cc}
I & 0\end{array}\right],\\
 & Lu^{i}=\left[\left(u_{1}^{i}\right)^{T}\left(u_{2}^{i}\right)^{T}\ldots\left(u_{M}^{i}\right)^{T}\right]^{T}
\end{eqnarray*}
The previous Optimal Control Problem (\ref{eq:rewrittenOCP}) is then
transformed into ADMM form,
\begin{eqnarray}
 & \underset{y_{1,}y_{2}}{\mbox{minimise}} & f_{1}\left(y_{1}\right)+f_{2}\left(y_{2}\right),\label{eq:ADMMa}\\
 & \mbox{subject to} & M_{1}y_{1}+M_{2}y_{2}=0,\label{eq:ADMMb}
\end{eqnarray}
with the optimisation variables defined as
\begin{eqnarray}
y_{1} & = & \left[\begin{array}{ccccc}
\check{u}^{T} & x^{T} & z^{T} & \check{\psi}^{T} & \check{\mu}\end{array}\right]^{T},\\
y_{2} & = & \left[\begin{array}{ccc}
u^{T} & \psi^{T} & \mu^{T}\end{array}\right]^{T}
\end{eqnarray}
where
\begin{eqnarray}
 & g=\left[\begin{array}{ccccc}
0 & 0 & 0 & 0 & \lambda\end{array}\right]^{T}, & H=\left[\begin{array}{ccc}
0 & 0 & 0\\
0 & \tilde{\lambda}I & -\tilde{\lambda}\mathbf{1}^{T}\\
0 & -\tilde{\lambda}\mathbf{1} & S\tilde{\lambda}
\end{array}\right],\\
 & M_{1}=\left[\begin{array}{ccccc}
0 & 0 & 0 & 0 & 1\\
0 & 0 & 0 & 0 & 1\\
I & 0 & 0 & 0 & 0\\
0 & 0 & 0 & I & 0
\end{array}\right], & M_{2}=\left[\begin{array}{ccc}
0 & \frac{-\mathbf{1}^{T}}{S} & 0\\
0 & 0 & -1\\
-I & 0 & 0\\
0 & -I & 0
\end{array}\right],
\end{eqnarray}
\begin{eqnarray}
f_{1}\left(y_{1}\right) & = & g^{T}y_{1}+I_{\mathbb{Y}_{1}}\left(y_{1}\right),\\
f_{2}\left(y_{2}\right) & = & y_{2}^{T}Hy_{2}+I_{\mathbb{Y}_{2}}\left(y_{2}\right),\\
\nonumber 
\end{eqnarray}
\begin{eqnarray}
\mathbb{Y}_{1} & = & \left\{ y_{1}\mid\tilde{A}x+\tilde{B}\check{u}+\tilde{w}=0,z=\tilde{C}x,\check{\psi}\geq\tilde{\phi}\left(\check{u},x,z\right)\right\} ,\\
\mathbb{Y}_{2} & = & \left\{ y_{2}\mid\tilde{L}u=0\right\} 
\end{eqnarray}

\subsection{Decentralised Iterated Computation}

The Lagrangian of (\ref{eq:ADMMa}) and (\ref{eq:ADMMb}) is
\begin{equation}
\mathcal{L}\left(y_{1},y_{2},\zeta\right)=f_{1}\left(y_{1}\right)+f_{2}\left(y_{2}\right)+\zeta^{T}\left(M_{1}y_{1}+M_{2}y_{2}\right)
\end{equation}
where $\zeta$ is a vector of Lagrangian multipliers for (\ref{eq:ADMMb}).
In ADMM, points satisfying the optimality conditions for (\ref{eq:ADMMa})
and (\ref{eq:ADMMb}) are obtained via the recursions with iteration
number $j$
\begin{eqnarray}
 &  & \begin{array}{c}
y_{1}\left(j+1\right)=\underset{y_{1}}{\mbox{arg min}}\mathcal{L}_{\rho}\left(y_{1},y_{2}\left(j\right),\zeta\left(j\right)\right)\\
=\underset{y_{1}}{\mbox{arg min}}f_{1}\left(y_{1}\right)+\frac{\rho}{2}\left\Vert M_{1}y_{1}+M_{2}y_{2}\left(j\right)+\eta\left(j\right)\right\Vert _{2}^{2}
\end{array}\label{eq:ADMMrecursionA}\\
 &  & \begin{array}{c}
y_{2}\left(j+1\right)=\underset{y_{2}}{\mbox{arg min}}\mathcal{L}_{\rho}\left(y_{1}\left(j+1\right),y_{2},\zeta\left(j\right)\right)\\
=\underset{y_{2}}{\mbox{arg min}}f_{2}\left(y_{2}\right)+\frac{\rho}{2}\left\Vert M_{1}y_{1}\left(j+1\right)+M_{2}y_{2}+\eta\left(j\right)\right\Vert _{2}^{2},
\end{array}\label{eq:ADMMrecursionB}\\
 &  & \eta\left(j+1\right)=\eta\left(j\right)+\left(M_{1}y_{1}\left(j+1\right)+M_{2}y_{2}\left(j+1\right)\right)\label{eq:ADMMrecursionC}
\end{eqnarray}
where the augmented Lagrangian with penalty parameter $\rho>0$ is
defined as
\[
\mathcal{L}_{\rho}\left(y_{1},y_{2},\zeta\right)=\mathcal{L}\left(y_{1},y_{2},\zeta\right)+\frac{\rho}{2}\left\Vert M_{1}y_{1}+M_{2}y_{2}\right\Vert _{2}^{2}
\]
and $\eta=\zeta/\rho$ is a scaled dual variable.

Stopping criteria for the previous recursions (\ref{eq:ADMMrecursionA}),
(\ref{eq:ADMMrecursionB}) and (\ref{eq:ADMMrecursionC}) is given
by 
\begin{eqnarray}
\left\Vert M_{1}y_{1}\left(j\right)+M_{2}y_{2}\left(j\right)\right\Vert _{2} & \leq & \varepsilon_{P},\\
\rho\left\Vert M_{1}^{T}M_{2}\left(y_{2}\left(j+1\right)-y_{2}\left(j\right)\right)\right\Vert _{2} & \leq & \varepsilon_{D},
\end{eqnarray}
indicating that the algorithm should be stopped when the optimality
conditions for (\ref{eq:ADMMa}) and (\ref{eq:ADMMb}) are satisfied
with accuracy as defined by the small tolerance levels $\varepsilon_{P}$
and $\varepsilon_{D}$.

The following Algorithm \ref{alg:ADMM-algorithm-1} describes the
steps of the implementation of the ADMM recursions (\ref{eq:ADMMrecursionA})-(\ref{eq:ADMMrecursionC}):
further optimisations are possible to parallelise the algorithm in
$S$.

\noindent \begin{flushleft}
\fbox{\begin{minipage}[t]{1\columnwidth}%
\begin{algorithm}[H]
\protect\caption{\label{alg:ADMM-algorithm-1}ADMM algorithm for the Optimal Control
Problem \ref{eq:optimalControlProblem}-\ref{eq:optimalControlProblemEnd}}

\textbf{while not} converged \textbf{do}

~~// ADMM update of $y_{1}=\left(\begin{array}{ccccc}
\check{u}^{T}, & x^{T}, & z^{T}, & \check{\psi}^{T}, & \check{\mu}\end{array}\right)$

~~$\left(\begin{array}{ccccc}
\check{u}^{T}, & x^{T}, & z^{T}, & \check{\psi}^{T}, & \check{\mu}\end{array}\right)\leftarrow$compute via \ref{eq:ADMMrecursionA}

\textbf{~~}// ADMM update of $y_{2}=\left(\begin{array}{ccc}
u^{T}, & \psi^{T}, & \mu^{T}\end{array}\right)$

~~$\left(\begin{array}{ccc}
u^{T} & \psi^{T} & \mu^{T}\end{array}\right)\leftarrow$compute via \ref{eq:ADMMrecursionB}

~~// ADMM update of $\eta$

~~$\eta\leftarrow$compute via \ref{eq:ADMMrecursionC}

\textbf{end while}
\end{algorithm}
\end{minipage}}
\par\end{flushleft}
\begin{thm}
\label{thm:faithful-stablecoins}The proposed decentralised mechanism
in Algorithm \ref{alg:ADMM-algorithm-1} is a faithful decentralised
implementation.\end{thm}
\begin{proof}
The steps that every rational user $i$ will faithfully complete are
the variable update steps of (\ref{eq:ADMMrecursionA})-(\ref{eq:ADMMrecursionC})
in Algorithm \ref{alg:ADMM-algorithm-1}.

Under the assumption of rational players in an ex-post Nash equilibrium
(\ref{def:Ex-Post-Nash-Equilibrium}), users can maximise their own
utility only by maximising the social welfare (\thmref{strategyproof-auction}).
Therefore, every user will faithfully execute the variable update
steps of (\ref{eq:ADMMrecursionA})-(\ref{eq:ADMMrecursionC}) since
it's the only way to maximise social welfare when all the other rational
users are following the intended strategy.
\end{proof}

\section{Auction Mechanism for Issuing Stablecoins\label{sec:Auction-Mechanism-Issuing-Stablecoins}}

At the beginning of an auction, each user reports its demand to the
auction manager. We define the demand of user $i$ as
\begin{equation}
\theta_{i}=\left\{ x_{i}^{min}\left(t\right),x_{i}\left(t\right),x_{i}^{max}\left(t\right)\right\} 
\end{equation}
Users can misreport their demands: let $\hat{\theta}_{i}=\left\{ \hat{x}_{i}^{min}\left(t\right),\hat{x}_{i}\left(t\right),\hat{x}_{i}^{max}\left(t\right)\right\} $
denote the reported demand of user $i$. The auction manager determines
the outcome of the auction including stablecoin allocation and payments
according to the stablecoin allocation rule, $al\left(\right)$.

Denote the following variable definitions 
\begin{eqnarray*}
 & x_{i}=\left[x_{i,1},\ldots,x_{i,T}\right], & y=\left[y_{1},\ldots,y_{T}\right],\\
 & v_{i}\left(x_{i}\right)=\sum_{t\in T}v_{i,t}\left(x_{i,t}\right), & c\left(y\right)=\sum_{t\in T}c_{t}\left(y_{t}\right).
\end{eqnarray*}
where $v_{i,t}\left(x_{i,t}\right)$ is a concave function for the
valuation of user $i$ at time $t$ and $c_{t}\left(y_{t}\right)$
is an always-positive convex function for the cost of the auction
manager at time $t$ (i.e., this cost is the market value of the auctioned
coins, plus other expenditures for carrying out the auction). The
utility of user $i$ is defined as the valuation minus the payment
\begin{equation}
u_{i}\left(al\left(\hat{\theta}\right),\theta_{i}\right)=\sum_{t\in T}v_{i,t}\left(x_{i},t\right)-\sum_{t\in T}p_{i,t}\left(\hat{\theta}\right),
\end{equation}
and the utility of the auction manager is the total payment minus
the total cost,
\begin{equation}
\sum_{i\in N}\sum_{i\in T}p_{i,t}\left(\hat{\theta}\right)-\sum_{t\in T}c_{t}\left(y_{t}\right)
\end{equation}

The stablecoin allocation rule of the auction mechanism is defined
by the following social welfare maximisation problem
\begin{eqnarray}
\mathcal{S}: & \underset{x,y}{\mbox{maximise}}\sum_{i\in N}v_{i}\left(x_{i}\right)-c\left(y\right),\label{eq:allocationRuleA}\\
\mbox{such that} & x_{i}\in X_{i}, & \forall i\in N,\label{allocaitonRuleB}\\
 & y\in Y,\label{allocationRuleC}\\
 & \sum_{i\in N}A_{i}x_{i}+By=0,\label{eq:allocationRuleD}
\end{eqnarray}
where $X_{i}$ is the constraint set of user $i$ for satisfying (\ref{eq:userConstraint})
$\forall t\in T$; $Y$ is the constraint set of the blockchain satisfying
(\ref{eq:blockchainFirstConstraint})-(\ref{eq:blockchainFirstHalfConstraint})
and (\ref{eq:blockchainSecondHalfConstraint})-(\ref{eq:blockchainLastConstraint})
$\forall t\in T$; $A_{i}$ and $B_{i}$ are the constraint set for
satisfying\textcolor{black}{{} (\ref{eq:auctionConstraint}) }$\forall t\in T$
equivalent to constraint (\ref{eq:allocationRuleD}).

The optimal solution to the social welfare maximisation problem $\mathcal{S}$
is denoted by $\left\{ x^{*},y^{*}\right\} $, in which $x^{*}$ is
the outcome of stablecoin allocation to users whenever all users truthfully
report their demands to the auction mechanism.

The payment by user $i$ at time slot $t$ is defined as the following
equation according to the VCG payment rule\cite{feasibleVCG,distributedVCG},
\begin{equation}
p_{i,t}\left(\theta\right)=\sum_{j\neq i}v_{j,t}\left(x_{j,t}^{-i}\right)-\sum_{j\neq i}v_{j,t}\left(x_{j,t}^{*}\right)+c_{t}\left(y_{t}^{*}\right),\label{eq:paymentAuction}
\end{equation}
where $x^{-i}=\left\{ x_{j,t}^{-i}\left|j\in N\setminus\left\{ i\right\} ,t\in T\right.\right\} $:
at the same time, the payment by user $i$ at time slot $t$ is the
optimal solution to the following maximisation problem that excludes
user $i$,
\begin{eqnarray}
\mathcal{S}_{-i}: & \underset{x}{\mbox{maximise}}\sum_{j\neq i}v_{j}\left(x_{j}\right),\label{eq:inverseEnergyAllocationRule}\\
\mbox{such that} & x_{j}\in X_{j}, & \forall j\in N\setminus\left\{ i\right\} 
\end{eqnarray}

\subsection{Properties of the Auction Mechanism}

In the proposed auction mechanism, each user achieves maximum utility
only when said user truthfully reports its demand $\theta_{i}$: a
mechanism is incentive-compatible if truth-revelation by users is
obtained in an equilibrium\cite{feasibleVCG,distributedVCG}.

Let $s_{i}\left(\theta_{i}\right)$ denote the strategy of user $i$
given $\theta_{i}$ and let 
\[
\theta_{-i}=\left\{ \theta_{1},\ldots,\theta_{i-1},\theta_{i+1},\ldots,\theta_{N}\right\} 
\]

\begin{defn}
(Dominant-Strategy Equilibrium\cite{overcomingManipulation}). A strategy
profile $s^{*}$ is a dominant-strategy equilibrium of a game if,
for all $i$,
\begin{equation}
u_{i}\left(f\left(s_{i}^{*}\left(\theta_{i}\right),s_{-i}\left(\theta_{-i}\right)\right),\theta_{i}\right)\geq u_{i}\left(f\left(s_{i}\left(\theta_{i}\right),s_{-i}\left(\theta_{-i}\right)\right),\theta_{i}\right)
\end{equation}
holds $s_{i}\left(\theta_{i}\right)\in\varTheta_{i}$, $\forall\theta_{i}$,
$\forall\theta_{-i}$ and $\forall s_{i}\neq s_{i}^{*}$.
\end{defn}
~
\begin{defn}
\label{def:(Strategy-Proof-Mechanism)}(Strategy-Proof Mechanism\cite{overcomingManipulation}).
A mechanism is strategy-proof if truthfully reporting demand $\theta_{i}$
is the best strategy of user $i$, no matter what the other users
report: that is, the incentive-compatibility of a mechanism in a dominant-strategy
equilibrium is only achieved when the following condition holds
\begin{equation}
u_{i}\left(f\left(\theta_{i},\hat{\theta}_{-i}\right),\theta_{i}\right)\geq u_{i}\left(f\left(\hat{\theta}_{i},\hat{\theta}_{-i}\right),\theta_{i}\right)
\end{equation}

\end{defn}
The proposed auction mechanism is incentive-compatible and strategy-proof
in a dominant-strategy equilibrium.
\begin{thm}
\label{thm:strategyproof-auction}The proposed auction mechanism (\ref{eq:allocationRuleA})-(\ref{eq:allocationRuleD})
and (\ref{eq:paymentAuction}) is strategy-proof.\end{thm}
\begin{proof}
We prove that each user will truthfully report their demand in order
to show that the auction mechanism is strategy-proof.

For their demanded amount $x_{i}(t)$, the payment rule (\ref{eq:paymentAuction})
was designed according to the VCG payment rule\cite{feasibleVCG,distributedVCG}
so that user's utility is maximised only when it truthfully reports
its demand.

For the lower bound $x_{i}^{min}\left(t\right)$, user $i$ will not
understate $x_{i}^{min}\left(t\right)$ to ensure that the minimum
demanded is satisfied. A user will not overstate $x_{i}^{min}\left(t\right)$
to avoid limiting the growth of the social welfare: to understand
the underlying reason, we write the utility of the user with the payment
rule expanded
\[
u_{i}\left(al\left(\hat{\theta}\right),\theta_{i}\right)=v_{i}\left(x_{i}^{*}\right)-\sum_{j\neq i}v_{j}\left(x_{j}^{-i}\right)+\sum_{j\neq i}v_{j}\left(x_{j}^{*}\right)-c\left(y^{*}\right),
\]
and note that a user cannot influence the second term by misreporting
their demand $\hat{\theta}$. A user maximising utility can only maximise
the other terms (i.e., social welfare). Therefore, user $i$ will
not overstate $x_{i}^{min}\left(t\right)$.

For the upper bound $x_{i}^{max}\left(t\right)$, for similar reasons
to the previous $x_{i}^{min}\left(t\right)$, understating $x_{i}^{max}\left(t\right)$
would only limit the growth of the social welfare, thus user $i$
is not incentivised to understate $x_{i}^{max}\left(t\right)$. On
the other side, overstating $x_{i}^{max}\left(t\right)$ would lead
to a larger stablecoin allocation than the real user's demand: the
auction manager would detect such a situation when later the user
is unable to pay the overstated allocation, and penalises the user
with much higher prices for a much lower amount of coins. Thus, user
$i$ will not overstate $x_{i}^{max}\left(t\right)$ in order to prevent
penalties.\end{proof}
\begin{thm}
The proposed auction is budget-balanced, that is, the received payment
is no less than the total cost.\end{thm}
\begin{proof}
The total payment that the auction manager receives is
\[
\sum_{i\in N}\sum_{i\in T}p_{i,t}\left(\theta\right)=\sum_{i\in N}\sum_{j\neq i}v_{j}\left(x_{j}^{-i}\right)-\sum_{i\in N}\sum_{j\neq i}v_{j}\left(x_{j}^{*}\right)+N\cdot c\left(y^{*}\right)
\]
Note that
\[
\sum_{j\neq i}v_{j}\left(x_{j}^{-i}\right)\geq\sum_{j\neq i}v_{j}\left(x_{j}^{*}\right)
\]
because $x^{-i}$ is the optimal solution to (\ref{eq:inverseEnergyAllocationRule})
and that, by definition, $c\left(y^{*}\right)\geq0$. Therefore, we
conclude 
\[
\sum_{i\in N}\sum_{i\in T}p_{i,t}\left(\theta\right)\geq N\cdot c\left(y^{*}\right)\geq c\left(y^{*}\right).
\]

\end{proof}

\section{Decentralised Implementation of Auction Mechanism\label{sec:Distributed-Implementation-Auction-Mechanism}}

A decentralised implementation of the centralised auction mechanism
(\ref{eq:allocationRuleA}) is achieved in this section: proximal
dual consensus ADMM\cite{ADMMdistrStats,proximalconsensusADMM} is
used to solve problem $\mathcal{S}$.

\subsection{Dual Consensus ADMM}

We start adding a polyhedra constraint to the stablecoin allocation
rule (\ref{eq:allocationRuleA})-(\ref{eq:allocationRuleD}):

\begin{eqnarray}
\mathcal{S}: & \underset{x,y}{\mbox{maximise}}\sum_{i\in N}v_{i}\left(x_{i}\right)-c\left(y\right),\label{eq:allocationRuleAPolyhedra}\\
\mbox{such that} & x_{i}\in X_{i}, & \forall i\in N,\label{allocaitonRuleBPolyhedra}\\
 & y\in Y,\label{allocationRuleCPolyhedra}\\
 & \sum_{i\in N}A_{i}x_{i}+By=0,\label{eq:allocationRuleDPolyhedra}\\
 & C_{i}x_{i}\preceq d_{i}, & i=1,\ldots,N,\label{eq:allocationRuleEPolyhedra}
\end{eqnarray}
where each $x_{i}$ in (\ref{eq:allocationRuleEPolyhedra}) is a local
constraint set of user $i$ consisting of simple polyhedra constraint
$C_{i}x_{i}\preceq d_{i}$, such that there would be closed-form solutions
to efficiently solve all the subproblems at every iteration.

Let $\lambda$ be the dual variable of constraint (\ref{eq:allocationRuleDPolyhedra}),
and $z_{i}$ be the dual variable of (\ref{eq:allocationRuleEPolyhedra}):
the Lagrange dual problem of $\mathcal{S}$, equivalent to solving
problem $\mathcal{S}$ since it's a concave maximisation problem,
is defined by
\begin{equation}
\underset{\lambda,z_{i}}{\mbox{minimise}}\sum_{i\in N}\phi_{i}\left(\lambda,z_{i}\right)+z_{i}^{T}d_{i}+\psi\left(\lambda\right),
\end{equation}
where
\begin{eqnarray}
\phi_{i}\left(\lambda,z_{i}\right) & = & \underset{x_{i}\in X_{i}}{\mbox{maximise}}\left\{ v_{i}\left(x_{i}\right)-\lambda^{T}A_{i}x_{i}-z_{i}^{T}\left(C_{i}x_{i}+r_{i}\right)\right\} ,\forall i\in N,\\
\psi\left(\lambda\right) & = & \underset{y\in Y}{\mbox{maximise}}\left\{ -c\left(y\right)-\lambda^{T}By\right\} 
\end{eqnarray}
where $r_{i}$ are slack variables. Let's obtain a copy of $\lambda$
for every user $i$, denoted by $\lambda_{i}$, by rewriting the previous
problem into the following equivalent problem,
\begin{eqnarray}
 & \underset{\lambda,z_{i},\left\{ \lambda_{i}\right\} ,\left\{ \lambda_{i}^{'}\right\} }{\mbox{minimise}}\sum_{i\in N}\phi_{i}\left(\lambda_{i},z_{i}\right)+z_{i}^{T}d_{i}+\psi\left(\lambda\right)\\
\mbox{such that} & \lambda_{i}=\lambda_{i}^{'}, & \forall i\in N,\\
 & \lambda=\lambda_{i}^{'},
\end{eqnarray}
In blockchain settings, there could be some users offline and/or some
communication links could be interrupted: at each iteration, each
user $i$ has probability $\alpha_{i}\in\left(0,1\right]$ of being
online, and each link $\left(i,j\right)$ has probability $p_{e}\in\left(0,1\right]$
of being interrupted; the probability that user $i$ and user $j$
are both active and able to exchange messages is given by $\beta_{ij}=\alpha_{i}\alpha_{j}\left(1-p_{e}\right)$.
For each iteration $k$, let $\Omega^{k}$ be the set of active users
and $\Psi^{k}\subseteq\left\{ \left(i,j\right)|i,j\in\Omega^{k}\right\} $
be the set of active edges.

The variable update steps of the auction manager at iteration $k$
are given by the following equations: 
\begin{eqnarray}
 &  & \mu^{\left[k\right]}=\mu^{\left[k-1\right]}+q\sum_{i\in N}\left(\lambda^{\left[k-1\right]}-\lambda_{i}^{\left[k-1\right]}\right),\label{eq:ADMMauctionManagerRecursionA}\\
 &  & \begin{array}{c}
y^{\left[k\right]}=\underset{y\in Y}{\mbox{arg min}}\left\{ c\left(y\right)+\frac{q}{4N}\left\Vert \frac{1}{q}By-\frac{1}{q}\mu^{\left[k\right]}\right.\right.\\
\left.+\sum_{i\in N}\left(\lambda^{\left[k-1\right]}+\lambda_{i}^{\left[k-1\right]}\right)\parallel_{2}^{2}\right\} ,
\end{array}\label{eq:ADMMauctionManagerRecursionB}\\
 &  & \lambda^{\left[k\right]}=\frac{1}{2N}\left(\frac{1}{q}By^{\left[k\right]}-\frac{1}{q}\mu^{\left[k\right]}+\sum_{i\in N}\left(\lambda^{\left[k-1\right]}+\lambda_{i}^{\left[k-1\right]}\right)\right)\label{eq:ADMMauctionManagerRecursionC}
\end{eqnarray}
with $\mu$ represents the dual variables $\lambda_{i}=\lambda_{i}^{'}$
and $q$ is a positive constant. The variable update steps of user
$i$ at iteration $k$ are given by the following equations:
\begin{eqnarray}
 & \forall i\in\Omega^{k}:\nonumber \\
 &  & \mu_{i}^{\left[k\right]}=\mu_{i}^{\left[k-1\right]}+2q\left(\lambda_{i}^{\left[k-1\right]}-t_{ij}^{\left[k-1\right]}\right),\label{eq:ADMMuserRecursionA}\\
 &  & \begin{array}{c}
\left(x_{i}^{\left[k\right]},r_{i}^{\left[k\right]}\right)=\underset{x_{i}\in X_{i},r_{i}\succ0}{\mbox{arg min}}\left\{ -v_{i}\left(x_{i}\right)+\frac{q}{4}\left\Vert \frac{1}{q}A_{i}x_{i}-\frac{1}{q}\mu_{i}^{\left[k\right]}\right.\right.\\
+2t_{ij}^{\left[k-1\right]}\parallel_{2}^{2}\\
\left.+\frac{1}{2\sigma_{i}}\left\Vert C_{i}x_{i}+r_{i}-d_{i}+\sigma_{i}z_{i}^{k-1}\right\Vert _{2}^{2}\right\} ,
\end{array}\label{eq:ADMMuserRecursionB}\\
 &  & z_{i}^{\left[k\right]}=z_{i}^{\left[k-1\right]}+\frac{1}{\sigma_{i}}\left(C_{i}x_{i}^{\left[k\right]}+r_{i}^{\left[k\right]}-d_{i}\right),\label{eq:ADMMuserRecursionC}\\
 &  & t_{ij}^{\left[k\right]}=\left\{ \begin{array}{cc}
\frac{\lambda_{i}^{\left[k\right]}+\lambda_{j}^{\left[k\right]}}{2}, & \mbox{if }\left(i,j\right)\in\Psi^{k},\\
t_{ij}^{\left[k-1\right]}, & \mbox{otherwise},
\end{array}\right.\label{eq:ADMMuserRecursionD}\\
 &  & \lambda_{i}^{\left[k\right]}=\frac{1}{2q}A_{i}x_{i}^{\left[k\right]}-\frac{1}{2q}\mu_{i}^{\left[k\right]}+t_{ij}^{\left[k-1\right]},\label{eq:ADMMuserRecursionE}\\
 & \forall i\notin\Omega^{k}:\nonumber \\
 &  & \begin{array}{c}
x_{i}^{\left[k\right]}\neq x_{i}^{\left[k-1\right]},r_{i}^{\left[k\right]}\neq r_{i}^{\left[k-1\right]},\lambda_{i}^{\left[k\right]}\neq\lambda_{i}^{\left[k-1\right]},z_{i}^{\left[k\right]}\neq z_{i}^{\left[k-1\right]},\\
\mu_{i}^{\left[k\right]}\neq\mu_{i}^{\left[k-1\right]},t_{ij}^{\left[k\right]}\neq t_{ij}^{\left[k-1\right]}\forall j\in N_{i},
\end{array}\label{eq:ADMMuserRecursionF}
\end{eqnarray}
where $\sigma_{i}$ are penalty parameters. The following stopping
criteria for the success of the convergence are applied by the auction
manager
\begin{eqnarray}
 &  & \left\Vert \lambda^{\left[k\right]}-\bar{\lambda}^{\left[k\right]}\right\Vert _{2}^{2}+\sum_{i\in N}\left\Vert \lambda_{i}^{\left[k\right]}-\bar{\lambda}^{\left[k\right]}\right\Vert _{2}^{2}\leq\varepsilon_{1},\label{eq:stoppingCriteriaAuctionA}\\
 &  & \left\Vert \bar{\lambda}^{\left[k\right]}-\bar{\lambda}^{\left[k-1\right]}\right\Vert _{2}^{2}\leq\varepsilon_{2},\label{eq:stoppingCriteriaAuctionB}
\end{eqnarray}
where $\varepsilon_{1}$ and $\varepsilon_{2}$ are small positive
constants and
\[
\bar{\lambda}^{\left[k\right]}=\left(\lambda^{\left[k\right]}+\sum_{i\in N}\lambda_{i}^{\left[k\right]}\right)/\left(N+1\right)
\]
The following Algorithm \ref{alg:ADMM-algorithm-2} shows the dual
consensus ADMM for problem $\mathcal{S}$:

\noindent %
\fbox{\begin{minipage}[t]{1\columnwidth}%
\begin{algorithm}[H]
\protect\caption{\label{alg:ADMM-algorithm-2}Dual Consensus ADMM for Problem $\mathcal{S}$}

$k=0$

Auction manager only: $\mu^{\left[0\right]}=0$, $y^{\left[0\right]}\in\mathbb{R}^{15T},\lambda^{\left[0\right]}\in\mathbb{R}^{3T}$

User $i$ only: $\mu_{i}^{\left[0\right]}=0$, $x_{i}^{\left[0\right]}\in\mathbb{R}^{15T},r_{i}^{\left[0\right]}\in\mathbb{R}^{15T},z_{i}^{\left[0\right]}\in\mathbb{R}^{15T},\lambda_{i}^{\left[0\right]}\in\mathbb{R}^{3T}$
and 
\[
t_{ij}^{\left[0\right]}=\frac{\lambda_{i}^{0}+\lambda_{j}^{0}}{2}
\]

\textbf{repeat}

~~$k\leftarrow k+1$

~~Auction manager only: send $\lambda^{\left[k-1\right]}$ to every
user $i$

~~Auction manager only: update $\mu^{\left[k\right]},y^{\left[k\right]}$
and $\lambda^{\left[k\right]}$ according to (\ref{eq:ADMMauctionManagerRecursionA})-(\ref{eq:ADMMauctionManagerRecursionC})

~~\textbf{for parallel $i\in N$ do}

~~~~User $i$ only: send $\lambda_{i}^{\left[k-1\right]}$ to
auction manager

~~~~User $i$ only: update $\mu_{i}^{\left[k\right]},x_{i}^{\left[k\right]}$,
$r_{i}^{\left[k\right]}$, $z_{i}^{\left[k\right]}$, $t_{ij}^{\left[k\right]}$
and $\lambda_{i}^{\left[k\right]}$ according to (\ref{eq:ADMMuserRecursionA})-(\ref{eq:ADMMuserRecursionF})

~~\textbf{end for}

\textbf{until }convergence is achieved by stopping criteria (\ref{eq:stoppingCriteriaAuctionA})
and (\ref{eq:stoppingCriteriaAuctionB});
\end{algorithm}
\end{minipage}}
\begin{thm}
\label{thm:ProximalADMM-convergence}Algorithm \ref{alg:ADMM-algorithm-2}
converges to the optimal solution of problem $\mathcal{S}$ in the
mean, with a $O\left(1/k\right)$ worst-case convergence rate.\end{thm}
\begin{proof}
Follows from Theorem 2 from \cite{proximalconsensusADMM}.
\end{proof}
Note that although this ADMM algorithm \ref{alg:ADMM-algorithm-2}
is only resistant against random failures $\alpha_{i}$ of users and
interruptions $p_{e}$ of the links, and not against poisoning attacks
that would corrupt inputs, it's also possible to design ADMM algorithms
resistant against Byzantine attackers: however, it would also increase
the number of iterations $k$, specially whenever under attack, thus
the chosen trade-off to ignore the Byzantine setting given the truthfulness
of \thmref{strategyproof-auction} and faithfulness of \thmref{faithful-distributed}
properties of the \nameref{sec:Distributed-Implementation-Auction-Mechanism}.

\subsection{Decentralised Mechanism\label{sub:Distributed-Mechanism}}

The decentralised mechanism features the following steps:

\noindent %
\fbox{\begin{minipage}[t]{1\columnwidth}%
\begin{center}
\label{Decentralised-Mechanism}\textbf{Protocol 1}: Decentralised
Mechanism of Auction
\par\end{center}
\begin{enumerate}
\item User $i$ reports his demand $\hat{\theta}_{i}$ to the auction manager.
\item User $i$ solves the following maximisation problem $\mathcal{S}_{i}$
\begin{equation}
x_{i}^{'}=\underset{x_{i}\in X_{i}}{\mbox{maximise }}v_{i}\left(x_{i}\right)
\end{equation}
and sends the result $x_{i}^{'}$ to the auction manager: since problem
$\mathcal{S}_{i}$ only requires local information, it can be solved
without collaborating with other users. The auction manager solves
problems $\mathcal{S}_{-i}$, $\forall i\in N$, by calculating
\begin{equation}
x^{-i}=\left\{ x_{j}^{'}\left|j\in N\setminus\left\{ i\right\} \right.\right\} 
\end{equation}
from the collected $x_{i}^{'}$, thus obtaining $\left\{ \mathcal{S}_{-1},\mathcal{S}_{-2},\ldots,\mathcal{S}_{-N}\right\} $.
\item To obtain the solution to problem $\mathcal{S},$ Algorithm \ref{alg:ADMM-algorithm-2}
is executed: the auction manager obtains results $y^{*}$ and $\lambda^{*}$,
and every user $i$ obtains $x_{i}^{*}$ and $\lambda_{i}^{*}$; every
user $i$ sends $x_{i}^{*}$ to the auction manager.
\item The auction manager calculates payments according to (\ref{eq:paymentAuction})
using the received $x^{*}$and $x^{-i}$, and obtains the stablecoin
allocation $x^{*}$.\end{enumerate}
\end{minipage}}

\subsection{Properties of the Decentralised Mechanism}

In the following, we prove that users will faithfully execute all
the actions of the \nameref{sub:Distributed-Mechanism} without manipulating
the outcome of the auction by strategically modifying results.
\begin{defn}
(Decentralised Mechanism \cite{distributedVCG}). A decentralised
mechanism $d_{M}=\left(g,\Sigma,s^{m}\right)$ defines an outcome
rule $g$, a feasible strategy space $\Sigma=\left(\Sigma_{1}\times\ldots\times\Sigma_{N}\right)$,
and an intended strategy $s^{m}=\left(s_{1}^{m},\ldots,s_{N}^{m}\right)$.
\end{defn}
~
\begin{defn}
(Intended Strategy \cite{distributedVCG}). A strategy $s^{m}$ is
the intended strategy of a decentralised strategy-proof direct-revelation
mechanism $M^{d}$ that implements outcome $f\left(\theta\right)$,
when
\[
f\left(\theta\right)=g\left(s^{m}\left(\theta\right)\right)
\]
for all $\theta\in\varTheta$.
\end{defn}
Thus, an intended strategy $s^{m}$ is a strategy that every user
is expected to follow: in the \nameref{sub:Distributed-Mechanism},
the intended strategies are all the steps that users must faithfully
execute to produce the same outcome as the centralised auction mechanism.
\begin{defn}
\label{def:(Faithful-Implementation)}(Faithful Implementation). A
decentralised mechanism $d_{M}=\left(g,\Sigma,s^{m}\right)$ is an
(ex-post) faithful implementation of social-choice rule $g\left(s^{m}\left(\theta\right)\right)$
when intended strategy $s^{m}$ is an ex-post Nash equilibrium.
\end{defn}
That is, users will follow the intended strategy in a faithful implementation
of a decentralised mechanism if no unilateral deviation can increase
their utility.
\begin{defn}
\label{def:Ex-Post-Nash-Equilibrium}(Ex-Post Nash Equilibrium \cite{distributedVCG,overcomingManipulation}).
A strategy profile $s^{*}=\left(s_{1}^{*},\ldots,s_{N}^{*}\right)$
is an ex-post Nash equilibrium when
\[
u_{i}\left(g\left(s_{i}^{*}\left(\theta_{i}\right),s_{-i}^{*}\left(\theta_{-i}\right)\right);\theta_{i}\right)\geq u_{i}\left(g\left(s_{i}^{'}\left(\theta_{i}\right),s_{-i}^{*}\left(\theta_{-i}\right)\right);\theta_{i}\right)
\]
for all agents, for all $s_{i}^{'}\neq s_{i}^{*}$, for every demand
$\theta_{i}$ and for all demands $\theta_{-i}$ of other agents.
\end{defn}
In an ex-post Nash equilibrium, all the other users are assumed rational:
thus, user $i$ will not deviate from $s_{i}^{*}$ when other users
are following strategy $s_{-i}^{*}$.
\begin{thm}
\label{thm:faithful-distributed}The proposed \nameref{sub:Distributed-Mechanism}
is a faithful decentralised implementation.\end{thm}
\begin{proof}
In the \nameref{sub:Distributed-Mechanism}, the steps that every
rational user $i$ will faithfully complete are the following:
\begin{enumerate}
\item Reporting $\hat{\theta}_{i}$ to the auction manager
\item Solving $\mathcal{S}_{i}$
\item Sending result $x_{i}^{'}$ of the previous step
\item Updating variable update steps $\mu_{i}^{\left[k\right]},x_{i}^{\left[k\right]}$,
$r_{i}^{\left[k\right]}$, $z_{i}^{\left[k\right]}$, $t_{ij}^{\left[k\right]}$
and $\lambda_{i}^{\left[k\right]}$of (\ref{eq:ADMMuserRecursionA})-(\ref{eq:ADMMuserRecursionF})
\item Sending $\lambda^{\left[k\right]}$ of (\ref{eq:ADMMuserRecursionE})
to the auction manager
\item Sending resulting $x_{i}^{*}$ obtained from the last step of (\ref{eq:ADMMuserRecursionB})
\end{enumerate}

Users will truthfully execute step 1 due to the truthful-revelation
property in a dominant-strategy equilibrium of Theorem \ref{thm:strategyproof-auction}
that also implies truthful-revelation in an ex-post Nash equilibrium.\\
Further, the calculation of $S_{i}$ is done locally without any input
from other users (i.e., the input from Byzantine attackers is never
considered) and the auction manager will only take a result $x_{i}^{'}$
from each identified user using a secure channel. Moreover, the computation
of $S_{i}$ does not solve problems $\mathcal{S}_{-i}$ and it cannot
modify the term $\sum_{j\neq i}v_{j}\left(x_{j}^{-i}\right)$ in the
payment rule (\ref{eq:paymentAuction}) (i.e., the user cannot lower
its payment). Thus, a rational user will faithfully execute steps
2 and 3.\\
Finally, users can maximise their own utility only by maximising the
social welfare, according to Theorem \ref{thm:strategyproof-auction}.
Therefore, every user will faithfully execute actions 4-6, since it's
the only way to maximise social welfare when all the other rational
users are following the intended strategy.

\end{proof}

\section{Encrypting ADMM\label{sec:Encrypting-ADMM}}

Previous works on encrypting ADMM or Model Predictive Control are
very scarce: there are some works about encrypting models from control
theory or model predictive control but only for cloud settings\cite{cloudEncryptedMPClinear,cloudMPCencryptedData,cloudQuadraticOptimization,cloudQuadraticOptimizationPHE,cloudSMPCcontrol},
thus non-decentralised; another paper encrypts ADMM models, but using
differential privacy\cite{ADMMdiffPrivacy}; yet another paper encrypts
ADMM models, but in the semi-honest setting\cite{ADMMdecentralize};
only Helen\cite{helen} encrypts ADMM in the malicious setting, thus
it will be our chosen framework .

Helen\cite{helen} solves a coopetive machine learning between multiple
parties in a malicious setting. Like other works where multiple parties
collaborate with their own data using secure multiparty computation\cite{incentivizeCollaboration},
they can't handle settings where the parties lie about their inputs
(i.e., poisoning attacks). One could argue that privacy only makes
lying worse: that is, privacy without truthfulness and faithfulness
is troublesome (\textit{Proverbs 12:22}, \cite{proverbs1222}). Fortunately,
the present paper solves all these issues by leaning on our previous
theorems about truthfulness of \thmref{strategyproof-auction} and
faithfulness of \thmref{faithful-distributed} for the \nameref{sec:Distributed-Implementation-Auction-Mechanism}.

\subsection{Cryptographic Gadgets}

We utilise the SPDZ framework\cite{cryptoeprint:2011:535}: an input
$a\in\mathbb{F}_{p^{k}}$ is represented as 
\[
\left\langle a\right\rangle =\left(\delta,\left(a_{1},\ldots,a_{n}\right),\left(\gamma\left(a\right)_{1},\ldots,\gamma\left(a\right)_{n}\right)\right)
\]
where $\delta$ is public, $a_{i}$ is a share of $a$ and $\gamma\left(a\right)_{i}$
is the MAX share authenticating $a$ under a SPDZ global key $\alpha$
that is not revealed until the end of the protocol. For an SPDZ execution
to be considered as correct, the following properties must hold
\[
\begin{array}{cc}
a=\sum_{i}a_{i}, & \alpha\left(a+\delta\right)=\sum_{i}\gamma\left(a\right)_{i}\\
\\
\end{array}
\]
From Helen\cite{helen}, we re-use the following gadgets:

\noindent %
\fbox{\begin{minipage}[t]{1\columnwidth}%
A zero-knowledge proof for the statement: ``Given public parameters:
public key $PK$, encryptions $E_{X}$, $E_{Y}$ and $E_{z}$; private
parameters $\mathbf{X}$,
\begin{itemize}
\item $Dec_{SK}\left(E_{Z}\right)=Dec_{SK}\left(E_{X}\right)\cdot Dec_{SK}\left(E_{Y}\right)$,
and
\item I know $\mathbf{X}$ such that $Dec_{SK}\left(E_{X}\right)=\mathbf{X}$''
\end{itemize}
\begin{center}
\textbf{\label{Gadget-1}Gadget 1}. Plaintext-ciphertext matrix multiplication
proof
\par\end{center}%
\end{minipage}}\\
\fbox{\begin{minipage}[t]{1\columnwidth}%
A zero-knowledge proof for the statement: ``Given public parameters:
public key $PK$, encryptions $E_{X}$, $E_{Y}$ and $E_{z}$; private
parameters $\mathbf{X}$ and $\mathbf{Y}$,
\begin{itemize}
\item $Dec_{SK}\left(E_{Z}\right)=Dec_{SK}\left(E_{X}\right)\cdot Dec_{SK}\left(E_{Y}\right)$,
and
\item I know $\mathbf{X},$$\mathbf{Y}$ and $\mathbf{Z}$ such that $Dec_{SK}\left(E_{X}\right)=\mathbf{X}$,
$Dec_{SK}\left(E_{Y}\right)=\mathbf{Y}$ and $Dec_{SK}\left(E_{Z}\right)=\mathbf{Z}$''
\end{itemize}
\begin{center}
\textbf{\label{Gadget-2}Gadget }2. Plaintext-plaintext matrix multiplication
proof
\par\end{center}%
\end{minipage}}\\
\fbox{\begin{minipage}[t]{1\columnwidth}%
For $m$ parties, each party having the public key $PK$ and a share
of the secret key $SK$, given public ciphertext $Enc_{PK}\left(a\right)$,
convert $a$ into $m$ shares $a_{i}\in\mathbb{Z}_{p}$ such that
\[
a\equiv\sum a_{i}\mbox{ mod }p
\]
Each party $P_{i}$ receives secret share $a_{i}$ and does not learn
the original secret value $a$.

\begin{center}
\textbf{\label{Gadget-3}Gadget }3. Converting ciphertexts into arithmethic
MPC shares
\par\end{center}%
\end{minipage}}\\
\fbox{\begin{minipage}[t]{1\columnwidth}%
Given public parameters: encrypted value $Enc_{PK}\left(a\right)$,
encrypted $SPDZ$ input shares $Enc_{PK}\left(b_{i}\right)$, encrypted
$SPDZ$ MACs $Enc_{PK}\left(c_{i}\right)$, and interval proofs of
plaintext knowledge, verify that:
\begin{enumerate}
\item $a\equiv\sum_{i}b_{i}\mbox{ mod }p$, and
\item $b_{i}$ are valid $SPDZ$ shares and $c_{i}$'s are valid MACs on
$b_{i}$.
\end{enumerate}
\begin{center}
\textbf{\label{Gadget-4}Gadget }4. MPC conversion verification
\par\end{center}%
\end{minipage}}

\subsection{Initialisation Phase}

During initialisation, the $m$ parties compute using SPDZ the parameters
for threshold encryption\cite{sharingDecryption}, generating a public
key $PK$ known to everyone. Each party $m$ receives a share of the
corresponding secret key $SK_{i}$: all the parties must agree to
decrypt a value encrypted with the shared $PK$.

\subsection{Input Preparation Phase}

In this phase, each party commits to their inputs by broadcasting
their encrypted inputs to all the other parties: additionally, all
the parties prove that they know the encrypted values using zero-knowledge
proofs of knowledge. Note that encryptions also serve as a commitment
scheme\cite{cryptoeprint:2009:007}.

To ensure that each party consistently uses the same inputs during
the entire protocol and to avoid deviations based on what other parties
have contributed, each party encrypts and broadcasts: $\mbox{Enc}_{PK}\left(\hat{\theta}_{i}\right)=\left\{ \hat{x}_{i}^{min}\left(t\right),\hat{x}_{i}\left(t\right),\hat{x}_{i}^{max}\left(t\right)\right\} $,
$\mbox{Enc}_{PK}\left(x_{i}^{'}\right)$, $\mbox{Enc}_{PK}\left(x_{i}\right)$
and $\mbox{Enc}_{PK}\left(y\right)$. These encryptions are accompanied
with proofs that the committed inputs are within a certain range\cite{cryptoeprint:2000}.

\subsection{Compute Phase}

In this phase, the variable update steps of the ADMM are executed,
in which parties successively compute locally on encrypted data, followed
by coordination steps with other parties using MPC computation. No
party learns any intermediate step beyond the final results, proving
in zero-knowledge that the local computations were performed correctly
using the data committed during the input preparation phase.

\subsubsection{Initialisation and Pre-Computations}

Initial variables are initialised to zero: $\mu^{\left[0\right]},\lambda^{\left[0\right]},\mu_{i}^{\left[0\right]},\lambda_{i}^{\left[0\right]},r_{i}^{\left[0\right]},z_{i}^{\left[0\right]},t_{ij}^{\left[0\right]}$.

Additionally, the auction manager solves problems $\mathcal{S}_{-i}$,
obtaining $x^{-i}$ from the collected $x_{i}^{'}$ in the preparation
phase.

\subsubsection{Local Optimisation}

Since Algorithm \ref{alg:ADMM-algorithm-2} is fully parallel and
decentralised, note that the variable update steps of auction manager
(\ref{eq:ADMMauctionManagerRecursionA})-(\ref{eq:ADMMauctionManagerRecursionC}),
or the steps (\ref{eq:ADMMuserRecursionA})-(\ref{eq:ADMMuserRecursionF})
of user $i$, only require local information and iterative exchange
of $\lambda^{\left[k\right]}$ and $\lambda_{i}^{\left[k\right]}$
with its neighbors.

Each party can independently calculate all the variable update steps
by doing plaintext scaling and plaintext-ciphertext matrix multiplication:
each party also needs to generate proofs proving that they have calculated
the variable update steps correctly, using Gadget 1 (\ref{Gadget-1})
and Gadget 2 (\ref{Gadget-2}).

\subsubsection{Coordination}

After the local optimisation step, each party exchanges $\lambda^{\left[k\right]}$
and $\lambda_{i}^{\left[k\right]}$ with its neighbors, and each party
also publishes interval proofs of knowledge. 

We may not need to use MPC: it's only required if steps (\ref{eq:ADMMauctionManagerRecursionB})
or (\ref{eq:ADMMuserRecursionB}) are implemented using non-linear
functions, which itself depends on the concrete functions $c\left(y\right)$
and $-v_{i}\left(x_{i}\right)$. In the best case, simple closed-form
solutions with only linear functions could be chosen.

But when MPC is needed, the encrypted variables need to be converted
to arithmetic SPDZ shares using Gadget 3 (\ref{Gadget-3}) and calculate
the function using SPDZ. After the MPC computation, each party receives
shares of the variables and its MAC shares: these shares are converted
back into encrypted form by encrypting the shares, publishing them,
and summing up the encrypted shares.

After all the ADMM calculation, every user $i$ sends $x_{i}^{*}$
to the auction manager, which must calculate payments according to
\eqref{paymentAuction} to obtain the stablecoin allocation $x^{*}$:
these calculations may also require MPC conversion and computation.

\subsection{Release Phase}

The encrypted model obtained at the end of the previous phase is decrypted:
all parties must agree to decrypt the results and release the final
data. Before said release, parties must prove that they correctly
executed the conversions between ciphertext and MPC shares using Gadget
4 (\ref{Gadget-4}), in order to prevent that different inputs from
the committed ones were used.

After all the SPDZ value have been verified by Gadget 4 (\ref{Gadget-4}),
the parties aggregate the encrypted shares of the stablecoin allocation
$x^{*}$ in to a single ciphertext, and then run the joint decryption
protocol\cite{cryptoeprint:2000}.

\subsection{Analysis of Properties}

Following the line of work merging secure computation and mechanism
design\cite{rationalSecureComputation}, that assumes that players
are rational and not only honest or malicious, we reach Guaranteed
Output Delivery (G.O.D.) and fairness\cite{cryptoeprint:2014:668},
circumventing their classical impossibility results.
\begin{defn}
$f_{CRS}:$ ideal functionality to generate common reference strings
and secret inputs to the parties.
\end{defn}
~
\begin{defn}
$f_{SPDZ}$: ideal functionality computing ADMM using SPDZ.\end{defn}
\begin{thm}
\label{thm:maliciousSecurity}$f_{DISTR-AUCTION-MECHANISM}$ is in
the $\left(f_{CRS},f_{SPDZ}\right)$-hybrid model under standard cryptographic
assumptions, against a malicious adversary who can statically corrupt
up to $m-1$ out of $m$ parties in an ex-post Nash equilibrium, reaching
G.O.D. and fairness, thus circumventing the impossibility results
of $f_{SPDZ}$.\end{thm}
\begin{proof}
Malicious security follows from Theorem 6 \cite{helen}.

The properties of truthfulness of \thmref{strategyproof-auction}
and faithfulness of \thmref{faithful-distributed} of the \nameref{sec:Distributed-Implementation-Auction-Mechanism},
imply that every rational party $i$ will faithfully complete all
the steps of $f_{DISTR-AUCTION-MECHANISM}$ : in other words, it won't
be rational to cheat or abort the protocol for malicious parties restricted
to the rational behaviors of an ex-post Nash equilibrium. Therefore,
we reach G.O.D. and fairness, thus their impossibility results are
circumvented.
\end{proof}

\section{Discussion}

The history of control theory for stabilisation in economics goes
back to the 1950s: for a recent survey, see \cite{controlTheoryEconomic}.
However, the ``Prescott critique''\cite{rulesRatherThanDiscretion,shouldControlTheory}
of the time-inconsistency of optimal control results precluded its
real-world applicability: fortunately, the problem of time inconsistency
can be adequately treated within the framework of Model Predictive
Control\cite{timeConsistentMPC}. And even though it might seem that
decentralising economic systems is a modern trend born from cryptocurrencies
and blockchains, there are already publications about these topics
starting from the 1970s: \cite{aokiDecentralizedStabilization,myokenOptimalStabilization,pindyckOptimal,neckControllability,neckStabilization,aokiDecentralized,neckDecentralized}.
This paper subsumes all these previous works because: 1) Model Predictive
Control provides a more expressive language to define economic policies;
2) the decentralisation provided by the ADMM decomposition allows
for more than the 2-3 parties previously considered) the mechanism
design techniques used in this paper guarantee more robust results.

Economists have recently created multiple models showing the benefits
of Centrally-Banked Digital Currencies (CBDC): said results also apply
to a CBDC implemented in the technical framework of a fully decentralised
cryptocurrency, as in the present paper. For example:
\begin{itemize}
\item Monetary transmission would strengthen\cite{broadeningNarrowMoney}.
\item A practical costless medium of exchange, and facilitate the systematic
and transparent conduct of monetary policy\cite{futureCBDC}.
\item Permanently raise GDP by as much as 3\%, due to reductions in real
interest rates, distortionary taxes, and monetary transaction costs;
and improve the ability to stabilise the business cycle\cite{macroeconomicsCBDC}.
\item Increases financial inclusion, diminishes the demand for cash, and
expands the depositor base of private banks\cite{impactCDBC}.
\item Address competition problems in the banking sector\cite{issueemoney}.
\end{itemize}
Common objections to the genuineness of decentralisation in stablecoins
are traversed here:
\begin{enumerate}
\item Need for centralised holding of funds: not by using other cryptocurrencies
as collateral. 
\item Auditors are required for verification: not by using zero-knowledge
proofs and other mathematical guarantees. 
\item Centralised price feeds: multiple verified agents could post the real-time
prices on the blockchain, or use an authenticated data feed for smart
contracts\cite{townCrier}. The issue of adversarial attacks to neural
networks is not relevant here because all price feeds are supposed
trustworthy.
\end{enumerate}
Finally, consensus-ADMM as described in this paper offers many advantages
over smart contracts running on replicated state machines (e.g., Ethereum):
\begin{enumerate}
\item Data intensive tasks such as deep-learning (\ref{sec:Decentralised-Currency-Prediction})
are nearly impossible to execute due to gas limits and storage costs.
\item Not all mining nodes would need to participate on the currency stabilisation
process: this special role could be reserved to a trustworthy subset
of nodes.
\item The lack of privacy in public permissionless blockchains renders algorithms
such as the decentralised auction (\ref{sec:Distributed-Implementation-Auction-Mechanism})
unfeasible to run.
\end{enumerate}

\section{Conclusion}

The present paper has tackled and successfully solved the problem
of designing a decentralised stablecoin with price stability guarantees
inherited from control theory (i.e., \nameref{sub:Closed-Loop-Stability})
and model predictive control (i.e., convergence of \thmref{ProximalADMM-convergence}).
Further guarantees required in a decentralised setting come from mechanism
design: truthfulness (definition \ref{def:(Strategy-Proof-Mechanism)},
\thmref{strategyproof-auction}) and faithfulness (definition \ref{def:(Faithful-Implementation)},
\thmref{faithful-distributed}, \thmref{faithful-stablecoins}). Additional
security against malicious parties of \thmref{maliciousSecurity}
is obtained from the combination of secure multi-party computation
and zero-knowledge proofs.

The flexibility of this framework including model predictive control,
which can accommodate a great variety of economic policies, combined
with the powerful predictive capabilities of artificial intelligence
techniques (e.g., neural networks and deep learning) foretell a whole
range of possibilities that will lead to better cryptocurrencies and
blockchains.

{\footnotesize{}\bibliographystyle{alpha}
\bibliography{bib}
}
\end{document}